\documentclass[letterpaper,twocolumn,pra,aps]{revtex4}
\newcommand\SIFT[0]{\operatorname{\mathrm{SIFT}}}
\newcommand\CTRL[0]{\operatorname{\mathrm{CTRL}}}
\newcommand\bket[1]{|\mathfrak{#1}\rangle}

\newcommand\bitz[0]{\mathfrak{0}}
\newcommand\bito[0]{\mathfrak{1}}

\usepackage{mathrsfs}
\usepackage{amsmath,amssymb,amsthm}
\newcommand\ket[1]{\,|#1\rangle}
\newcommand\fet[1]{\,|#1\rangle}

\newcommand\Span[0]{\operatorname{Span}}
\newcommand\braket[2]{\langle #1\mid#2\rangle}

\IfFileExists{fmsfutm.fd}{%
\usepackage{fourier}\usepackage{times}%
\newcommand\splus[0]{+}%
\newcommand\sminus[0]{-}%
}{\usepackage{mathptmx}%
\newcommand\splus[0]{\text{\normalfont{\raisebox{-0.15ex}{\texttt{+}}}}}%
\newcommand\sminus[0]{\text{\normalfont{\raisebox{-0.15ex}{\texttt{-}}}}}%
}
\newcommand\spm[0]{\pm}
\usepackage[cm]{fullpage}
\newtheoremstyle{note}%
{3pt}%
{3pt}%
{\itshape}%
{}%
{\itshape}%
{\ ---}%
{2pt}%
{\thmnote{#3}}%
\theoremstyle{note}
\newtheorem{lemma}{Lemma}
\newcommand\fsection[1]{\textit{#1\,---}}
\newcommand\fsubsection[1]{\textit{#1.}}

\usepackage{color}
\definecolor{mbcolor}{rgb}{0.99,0.14,0.02}
\definecolor{weirdcolor}{rgb}{0.5,0.00,1.00}

\newcommand\remove[1]{\relax}

\usepackage{color}
 \definecolor{earth}{rgb}{0.631,0.165,0.161}
\usepackage[colorlinks,hyperindex,citecolor=earth,linkcolor=earth,urlcolor=earth,bookmarks=true,bookmarksopen,pdfstartview=FitH]{hyperref}
\begin{document}
\author{Michel Boyer$^{1}$,  Tal Mor$^2$\\
\small 1. D\'epartement IRO, Universit\'e de Montr\'eal,
  Montr\'eal (Qu\'ebec) H3C 3J7, {Canada}. \\
\small 2. Computer Science Department, Technion, 
  Haifa 32000, {Israel}.}

\title{On the Robustness of (Photonic)
Quantum Key Distribution with Classical
Alice}
\date{\today}

\begin{abstract}
Quantum Key Distribution (QKD) with classical Bob has recently been suggested
and proven robust.
Following this work, QKD with classical Alice was also suggested and
proven robust. The above protocols are ideal in the sense that they make
use of qubits.
However, in the past, well-known
QKD protocols that were proven
robust and even proven unconditionally secure, when qubits are used,
were found to be totally insecure when
photons are used. This is due to sensitivity to photon losses (e.g.,
Bennett's two-state protocol)
or sensitivity to losses combined with multi-photon states (e.g.,
the photon-number-splitting attack on the weak-pulse
Bennett-Brassard protocol, BB84).
Here we prove that 
QKD with classical Alice is still robust when photon
losses and even multi-photon states are taken into account.
\end{abstract}

\maketitle

\fsection{Introduction}%
A two-way Quantum Key Distribution (QKD) protocol 
in which one of the parties (Bob)
uses only classical operations was recently introduced~\cite{boyer:140501}. 
A very interesting extension in which the originator 
always sends the same state 
$\ket{\splus}=(\bket{0}+\bket{1})/\sqrt{2}$ (see \cite{Note4}),
while in~\cite{boyer:140501} all four states, $\bket{0}$,
$\bket{1}$, $\ket{\splus}$ and 
$\ket{\sminus}=(\bket{0}-\bket{1})/\sqrt{2}$, are sent,
is suggested by Zou \textit{et al.}~\cite{zou:052312}. 
In both those ``semi-quantum'' key distribution (SQKD) protocols 
the qubits  
go from the originator Alice to (classical) Bob and back to Alice.
Bob 
either reflects a received qubit without touching its state (CTRL), 
or measures it in the standard (classical) basis and sends back his 
result as $\bket{0}$ or $\bket{1}$ (SIFT).

Following~\cite{Boyer-Mor-QPH2010} we prefer to call 
the originator 
in~\cite{zou:052312}  
Bob (and not Alice), and to call the classical party Alice:
usually in quantum cryptography, Alice is
the sender of some non-trivial data, 
e.g., she is the one choosing the quantum states. 
The originator in~\cite{zou:052312} does not have that special role, 
as the state $\ket{\splus}$ is always sent 
(and we could even ask Eve to generate it). 
The classical person is then the one actually choosing a basis
and knowing which of the three state ($\bket{0}$, $\bket{1}$, or $\ket{\splus}$) is
sent back to the originator, thus it   
is natural to name that classical person Alice.
We call the originator Bob, 
and we call the SQKD protocol of Zou et al ``QKD with classical Alice''.
Note that QKD with classical Alice was also suggested, independently
of~\cite{zou:052312}, by
Lu and Cai~\cite{LC2008}.
As proven in~\cite{Boyer-Mor-QPH2010},
QKD with Classical Alice (the protocol suggested in~\cite{zou:052312}) 
is  completely
robust against eavesdropping. 

Here we use Fock-space representation to
extend the QKD with classical Alice protocol to the important case
in which 
Alice and Bob use photons and not merely ideal qubits. 
We first extend the proof of robustness to include photon loss, 
and subsequently, also multi-photon states.
Such 
extensions are far from trivial; on the contrary, often, robustness is
actually lost when trying to deal with photons rather than qubits.

As a first example, in the two-state scheme   
(known as the Bennett'92 --- B92 scheme), when qubits
are assumed to be carried by photons, photon losses cause a severe problem: 
if %
Eve can replace a lossy channel by a lossless one, she might be able to get
full information without causing errors at all, using an ``un-ambiguous state
discrimination'' attack. 
See %
appendix, sect \ref{sectA}.
As a second example, in the four-state scheme  
(known as the 
Bennett-Brassard'84 ---  BB84 scheme), when qubits
are assumed to be carried by photons, photon losses combined with multi-photon
pulses cause a severe problem: 
if %
Eve can replace a lossy channel by a lossless one, and can measure photon
numbers (via a non-demolition measurement), she might be able to get
full information without causing errors at all, using a ``photon number
splitting'' attack.
See %
appendix, sect. \ref{sectpns}.

\fsection{The Fock space notations}%
The Fock space notations that serve as an extension of a qubit are as follows:
in the standard ($z$) basis, 
the Fock basis vector $\fet{0,1}$ stands for a single photon
in a qubit-state $\bket{0}$
and the Fock basis vector $\fet{1,0}$ stands for a single photon
in a qubit-state $\bket{1}$.
Naturally, the Hadamard ($x$) 
basis qubit-states are given by the superposition  
of those Fock states so that 
$[\fet{0,1} \pm \fet{1,0}]/\sqrt{2}$ stand for a single photon
in a qubit-state $\ket{\spm}=(\bket{0} \pm \bket{1})/\sqrt{2}$. 
The general state of this photonic qubit can then be written as 
$\alpha \fet{0,1} + \beta \fet{1,0}$, with $|\alpha|^2+|\beta|^2=1$.   

This photonic qubit lies in a much larger space called Fock space.
The first natural extension is  
$\fet{0,0}$ that describes the lack of photons (the vacuum state), 
a case of great practical importance,
as it enables dealing properly with photon loss. 
The next extension of a very high practical importance is that 
$\fet{2,0}$ describes two (indistinguishable) photons in the 
same qubit-state $\bket{1}$, 
$\fet{0,2}$ describes two (indistinguishable) photons in the 
same qubit-state $\bket{0}$, 
and $\fet{1,1}$ describes two (in this case, distinguishable) photons,
one in the 
qubit-state $\bket{0}$, 
and one in the 
qubit-state $\bket{1}$. This case (a six dimensional space, describing two or
less photons) 
was found very important in the photon number splitting attack~\cite{BLMS00},
as prior to that analysis, experimentalists assumed that the only impact
of high loss rate is on the bit-rate and not on security.

In general, if a single photon can be found in two orthogonal states
(these are called ``modes'' when discussing photons), then
$\fet{n_\bito,n_\bitz}$ represents $n_\bito$ (respectively $n_\bitz$) indistinguishable
photons in a qubit-state $\bket{1}$ (resp. $\bket{0}$). 
The numbers $n_\bitz$ and $n_\bito$ are then called the
occupation numbers of the two modes.
From now on, the notations $\bket{0} \equiv \ket{0,1}$, $\bket{1} \equiv \ket{1,0}$, $\ket{\splus} = (\ket{0,1}+\ket{1,0})/\sqrt{2}$
and $\ket{\sminus} = (\ket{0,1}-\ket{1,0})/\sqrt{2}$ will be used interchangeably. 
Similarly, since the single photon can also be found in 
$\ket{0,1}_x \equiv \ket{\splus}$ and
$\ket{1,0}_x \equiv \ket{\sminus}$ (namely, the $x$ basis), then 
$\fet{n_\sminus,n_\splus}$ represents 
$n_{\sminus}$ 
(resp. $n_{\splus}$) indistinguishable photons in qubit-state 
$\ket{\sminus}$ (resp. $\ket{\splus}$). 

More generally, one may consider more than two modes. For instance, the four
modes $\fet{n_{\bito b},n_{\bito a},n_{\bitz b},n_{\bitz a}}$ are the generalization of qu-quadrit
(say a photon in one of two arms $a$ or $b$, and one of two orthogonal
polarizations, denoted $\bitz$ or $\bito$). 

\fsection{The classical Alice protocol, dealing with losses}%
The originator Bob sends Alice qubits in the state 
$\ket{\splus}$ and
keeps in a quantum
memory all qubits he received back from her~\cite{Note2}. When $N$ qubits have been
sent and received, (classical)  
Alice announces 
publicly which qubits she reflected (without disturbing them); 
the originator Bob  
then checks that he received  
$\ket{\splus}$ and not 
$\ket{\sminus}$ 
on those positions (CTRL).
For the (SIFT) qubits measured by Alice in the standard (classical) 
$\{\bket{0};\bket{1}\}$ basis, 
a sample is chosen to be checked
for errors (TEST). 
The remaining SIFT bits serve for obtaining a final, secure key,
via error correction and privacy amplification, as in any conventional 
QKD protocol.
 
\fsubsection{Defining the (limited) ``photonic QKD with classical Alice'' protocol} %
The qubits  
are embedded in the 3-dimensional, 2-mode Fock space 
containing the qubit states $\ket{1,0}$ and $\ket{0,1}$ and the vacuum state
$\ket{0,0}$.
The Hilbert space describing Alice+Bob states is (for now)
the subspace
\begin{equation}\label{spacelosses}
\mathscr{H}_{AB} = \Span\big(\ket{0,1},\ket{1,0},\ket{0,0}\big) \subseteq
\mathscr{F}
\end{equation}
of the more general 2-mode Fock space  
($\mathscr{F}$).

In this photonic protocol, Bob is always sending the $\ket{\splus}$ state. %
Losses or vacuum states are modeled by the state $\ket{0,0}$, and thus, we must
define Alice's and Bob's operations when such states occur. 
Losses normally come from the interaction with the environment; 
as usual, the (worst case) analysis gives Eve total control
on the environment.
Classical Alice can  
either SIFT or CTRL~\cite{zou:052312,Boyer-Mor-QPH2010}. 
In the SIFT mode, Alice's ``measurement'' is described
(WLG)
with the adjunction of a probe,
extending $\mathscr{H}_{AB}$ to $\mathscr{H}_A \otimes \mathscr{H}_{AB}$, 
a unitary transformation and a measurement of her probe 
in the standard basis. Such a description
is meant to match the general framework of measurements in quantum information, and
may not correspond to the actual physical measurement performed by Alice.
Using the Fock-space notations, 
it is assumed that Alice adds a two-mode probe in a state $\ket{0,0}_A$ 
to get the state $\ket{0,0}_A\ket{\splus}_{AB}$.
Alice then performs one of the following
two operations (with $\ket{n_{\bito},n_{\bitz}}_{AB}$ in the $z$, i.e. the standard basis):
\begin{align}
U_{\CTRL}\ket{0,0}_A\ket{n_{\bito},n_{\bitz}}_{AB} 
&= \ket{0,0}_A\ket{n_{\bito},n_{\bitz}}_{AB} \label{donothing}\\
U_{\SIFT}\ket{0,0}_A\ket{n_{\bito},n_{\bitz}}_{AB} 
&= \ket{n_{\bito},n_{\bitz}}_A \ket{n_{\bito},n_{\bitz}}_{AB} \label{premeasure}
\end{align}
then she measures her probe in the standard classical basis and
sends Alice+Bob's state to Bob; 
in the case described by Eq.~\eqref{donothing} (CTRL) 
she needs not measure, 
still the probe and its measurement are added there only
to make the description uniform; 
Bob's original state ($\ket{\splus}_{AB}$) is reflected back to him, undisturbed.
In the case described 
by Eq.~\eqref{premeasure} (SIFT), Alice gets the outcome 
$n_{\bito}n_{\bitz}$, and the state 
$\ket{n_{\bito},n_{\bitz}}_{AB}$ is sent to Bob.
Note that, in order to analyze the enlarged space of the protocol, 
we {\em had to}
add the definition of Alice's operation on the added state,
$\ket{0,0}_{AB}$.  
Our choice of 
$U_{\SIFT}\ket{0,0}_A\ket{0,0}_{AB} 
= \ket{0,0}_A \ket{0,0}_{AB} $ is the most natural way of extending 
Alice's SIFT operation, and it thus becomes part of our definition of the
protocol ``Photonic-QKD with classical Alice''.

Naturally, 
when Bob measures in the classical ($z$) basis, 
he also measures the same three
states as Alice, $\ket{n_{\bito},n_{\bitz}}$ with $n_{\bitz}+n_{\bito} \le 1$.
However, the space $\mathscr{H}_{AB}$ \eqref{spacelosses}  
is also spanned by the orthonormal 
basis $\{\ket{\splus}$, $\ket{\sminus}$, $\ket{0,0}\}$, thus  
Bob (who is not limited to being classical) can perform a
measurement in this generalized $x$ basis of the qutrit. 

\fsubsection{Eve's attack on the (photonic) classical Alice protocol} %
Eve performs her attack in both directions; 
from Bob to Alice, Eve applies $U$; from Alice to Bob, Eve applies $V$.
We may assume, WLG, that Eve is using a fixed probe space $\mathscr{H}_E$ 
for her attacks in both directions.
The attack from Bob to Alice produces a state of the form
$ \ket{E_{01}}\ket{0,1}_{AB} +\ket{E_{10}}\ket{1,0}_{AB} +
\ket{E_{00}}\ket{0,0}_{AB} $ (namely  
$\sum_{n_{\bito},n_{\bitz}\ |\ n_{\bitz}+n_{\bito} \le 1}
\ket{E_{n_{\bito}n_{\bitz}}}
\ket{n_{\bito},n_{\bitz}}_{AB}
\in \mathscr{H}_E\otimes\mathscr{H}_{AB}
$),   
where the $\ket{E_{ij}}$ are non normalized 
(and potentially non-orthogonal)
vectors in $\mathscr{H}_E$. 
With Alice's probe attached we obtain
\begin{equation}
 \Psi = \sum_{n_{\bito},n_{\bitz}\ |\ n_{\bitz}+n_{\bito} \le 1}
\ket{E_{n_{\bito}n_{\bitz}}} \ket{0,0}_A
\ket{n_{\bito},n_{\bitz}}_{AB} \ ,
\end{equation}
in 
$\mathscr{H}_E\otimes \mathscr{H}_A \otimes \mathscr{H}_{AB}$.
In particular, if Eve does nothing then 
$\ket{E_{10}} = \ket{E_{01}} = \ket{E_{00}} \equiv \ket{E}$
and the state in Alice+Eve's hands, prior to Alice's operation, is 
$\ket{E}\ket{0,0}_A\ket{\splus}_{AB}\ $.

Going back to the general case, if Alice applies $U_{\CTRL}$, 
then the state in Eve+Alice hands (after Alice's CTRL action) is still $\ket{\Psi}$.
However, if Alice applies $U_{\SIFT}$, 
the resulting global state in Eve+Alice's hands is
\[
\sum_{n_{\bito},n_{\bitz}\ |\ n_{\bitz}+n_{\bito} \le 1}
\ket{E_{n_{\bito}n_{\bitz}}}\ket{n_{\bito},n_{\bitz}}_A\ket{n_{\bito},n_{\bitz}}_{AB}
\]
and after Alice has measured 
her probe, she gets some output ($\{00,01,10\}$), 
and some (non normalized) residual state that she sends back to Bob. 

Once Alice has performed her measurements and sent $\ket{i,j}_{AB}$ back to Bob via
Eve, the resulting global state (fully in Eve's hands)  is
\begin{center}
\begin{tabular}{|c|l|}\hline
\textit{Measurement} & \textit{State}\\ \hline\hline
$00$ & $ \ket{\psi_{00}} = \ket{E_{00}}\ket{0,0}_{AB}$ \\ \hline
$01$ & $ \ket{\psi_{01}} = \ket{E_{01}}\ket{0,1}_{AB}$ \\ \hline
$10$ & $ \ket{\psi_{10}} = \ket{E_{10}}\ket{1,0}_{AB}$ \\  \hline
CTRL & $\ket{\psi} = \ket{\psi_{00}}+\ket{\psi_{01}}+\ket{\psi_{10}}$ \\ \hline
\end{tabular}
\end{center}
where $\ket{\psi_{ij}}$ are not normalized, and where the $\ket{E_{ij}}$ were chosen by Eve. 
Eve now applies a unitary $V$ on $\mathscr{H}_E\otimes\mathscr{H}_{AB}$ 
and then
sends Bob his part of the resulting state.

\fsubsection{A proof of robustness} %
For Eve to stay undetectable, if Alice measured $\ket{0,0}$ 
(namely, the outcome $00$) in the SIFT mode, then Bob should have 
a probability zero of measuring $01$ or $10$, 
thus, a probability zero of receiving the states $\ket{0,1}$ or $\ket{1,0}$.
Similarly if Alice measured $10$ ($01$), 
then Bob should have a probability zero of measuring $01$ ($10$);
he could however get a loss, $00$.
The resulting (non normalized) Eve+Bob residual states thus take the form
$\ket{\psi_{00}'} = V \ket{\psi_{00}}  =
 \ket{H_{00}}\ket{0,0}_{AB} $ when a loss arrives, and otherwise, 
\begin{align}
\ket{\psi_{01}'} &= V \ket{\psi_{01}}   = \ket{F_{01}}\ket{0,1}_{AB} +
\ket{H_{01}}\ket{0,0}_{AB} \nonumber \\
\ket{\psi_{10}'} &= V \ket{\psi_{10}}  =
 \ket{F_{10}}\ket{1,0}_{AB} +  \ket{H_{10}}\ket{0,0}_{AB} \ . \label{got10}
\end{align}
Finally, $V$ being linear, the (normalized) residual state if Alice applied CTRL is 
$\ket{\psi'} \equiv V\ket{\psi} =
\ket{\psi_{00}'}+\ket{\psi_{01}'}+\ket{\psi_{10}'}$.

In order to check CTRL bits, Bob measures $\ket{\psi'}$
in the $x$ basis and checks if he gets a photon in the illicit state $\ket{\sminus}$. 
To avoid that, Eve must make sure that the overlap between Eve-Bob's state 
$\ket{\psi'}$ and Bob's state 
$ \ket{\sminus}$ is zero. 
This results with another limitation on Eve's attack:
the norm of ${}_{AB}\langle{\sminus}\big|\big(\ket{F_{01}}\ket{0,1}_{AB}\big)
+ {}_{AB}\langle{\sminus}\big|\big(\ket{F_{10}}\ket{1,0}{}_{AB}\big)$
must be $0$; 
namely, $\ket{F_{01}}\braket{\sminus}{0,1} + \ket{F_{10}}\braket{\sminus}{1,0} 
= (\ket{F_{01}}-\ket{F_{10}})/\sqrt{2}=0$, i.e.
$\ket{F_{01}} = \ket{F_{10}} = \ket{F}$ for some (non normalized) state $\ket{F}\in \mathscr{H}_E$.
The final global states \eqref{got10}  
if Alice measured $01$ and $10$ are thus (respectively) 
\begin{align}
 \ket{F}\ket{0,1}_{AB} &+ \ket{H_{01}}\ket{0,0}_{AB} \nonumber \\ 
 \ket{F}\ket{1,0}_{AB} &+ \ket{H_{10}}\ket{0,0}_{AB} \label{final-table}
\ , 
\end{align}
and if Bob does not get a loss, 
Eve's final state is $\ket{F}$ whether Bob measures $\ket{0,1}$ i.e., 
the bit $\bitz$, or $\ket{1,0}$, i.e., the bit $\bito$.
Eve's final probe is, thus, 
independent of all of Alice's and Bob's measurements, and is unentangled with their state.

Eve can thus get no information on %
the bits Alice and Bob agree upon without being detectable.
That reasoning can be done inductively bitwise to get robustness with $N$ qubits.

\fsection{The classical Alice protocol, dealing with losses and multi-photon
pulses}%
In practice, there are not just losses: %
when qubits are encoded using photon pulses, there may be more than
one photon per pulse, giving the eavesdropper more 
tools to get information on the SIFT bits.
We now allow the Hilbert space to contain all photonic states of the
above-mentioned two modes.
Namely, we consider all states $\ket{n_{\bito},n_{\bitz}}$ with $n_{\bitz}+n_{\bito} \ge 0$.
As before, we {\em must} %
specify Alice's and Bob's
operations on those states.

\fsubsection{Defining the (full) ``photonic QKD with classical Alice'' protocol} %
If Alice and Bob can distinguish one from more than one photon, extending the
results of the earlier section is rather trivial;
in brief, Eve becomes limited to the same space as in the previous section, or else she will 
be noticed.

The interesting extension is when Alice and Bob are limited, and cannot
tell %
a single photon pulse from a multi-photon pulse. 
It is conventional to say that they have ``detectors'' and not ``counters''.
This, of course, is in contrast to Eve who has counters, and who can do
whatever physics allows.

We now assume a specific realization of the Fock states, to make the limitation
on the measurements more clear.
We assume that the two classical states, $\bket{0}$
and $\bket{1}$, describe two pulses on the same arm, such that the photon
can either be in one pulse, in the other, or in a superposition such as
the (non-classical) state $\ket{\splus}$. 
Measurements are applied onto the two modes separately, 
using two detectors, 
thus a state $\ket{1,1}$ as well
as any state $\ket{n_{\bito},n_{\bitz}}$ with both $n_{\bito}\geq 1$ and $n_{\bito}\geq 1$ %
can be identified
as an error.
That will be enough to guarantee robustness.

As before, we assume that Alice's CTRL operation is given by
Eq.~\eqref{donothing}, yet now, with  
$n_{\bitz}$ and $n_{\bito}$ being any non-negative
integers. %
Let $\hat{n}_{\bito} = 1$ if $n_{\bito}\geq 1$, else $\hat{n}_{\bito}=0$; similarly, $\hat{n}_\bitz = 1$ if $n_{\bitz}\geq 1$, else
$\hat{n}_{\bitz} = 0$.
To model properly the use of a detector that clicks when noticing one or more
photons, it is assumed that in the SIFT mode Alice still attaches a probe
in the $\ket{0,0}_A$ state. Now she applies the following transform, $U_{\SIFT}$, 
on $\mathscr{H}_A\otimes \mathscr{H}_{AB}$ where $\mathscr{H}_A = 
\Span\big(\ket{0,0}_A, \ket{0,1}_A, \ket{1,0}_A, \ket{1,1}_A\big)$ 
and $\mathscr{H}_{AB}$ is $\mathscr{F}$, Alice+Bob's 
2-mode photonic space:
\begin{equation}\label{measurement}
U_{\SIFT}\ket{0,0}_A\ket{n_{\bito},n_{\bitz}}_{AB} =
\ket{\hat{n}_{\bito},\hat{n}_{\bitz}}_A\ket{n_{\bito},n_{\bitz}}_{AB} \ .
\end{equation}
Alice then measures her probe in the $\ket{0,0}_A$, 
$\ket{0,1}_A$, $\ket{1,0}_A$ and $\ket{1,1}_A$ basis;
she cannot distinguish 
$\ket{n_{\bito},0}$ with $n_{\bito} \ge 2$ 
from 
$\ket{1,0}$, yet she can distinguish 
$\ket{1,1}$  from  $\ket{1,0}$.  
When $n_{\bito}\geq 1$ or $n_{\bitz}\geq 1$ she sees $\hat{n}_{\bito}=1$ or $\hat{n}_{\bitz}=1$ (respectively);
if both
$n_{\bito}\geq 1$ and $n_{\bitz}\geq 1$ then she measures her probe in a state 
$\ket{1,1}_A$; this is telling her that the state she received is illicit.

We need to {\em carefully} define Alice's operation on the states she receives,
as the robustness analysis depends on the residual state after Alice's
``measurement'', which Alice sends back to Bob;
we now consider two legitimate options for defining
that state. 
In one, which we could call ``the conventional measure-resend approach'', 
we assume that depending on which detector clicks, the state
$\ket{0,1}$ or the state
$\ket{1,0}$ (or the state
$\ket{0,0}$ if no detector clicked) is then sent back to Bob.
However, now Eve could prepare the state 
$(\ket{0,2} + \ket{2,0})/\sqrt{2}$ and send it to Alice; in CTRL mode the
same 
state will return to Eve, while in SIFT mode only a single photon (or none) will
be given back to Eve.  Thus Eve (who can measure the number of photons)
will easily decode Alice's operation, and will be able 
to measure (and resend) in case of SIFT, or send the state 
$(\ket{0,1} + \ket{1,0})/\sqrt{2}$ back to Bob in case of CTRL.

We thus stick here to a different way of defining the residual state after Alice's
action: we simply assume that the state 
$\ket{n_{\bito},n_{\bitz}}$ is sent back to Bob in both Eq. \eqref{measurement} and Eq. \eqref{donothing}.
Incidently, that attack above is an example of a simple tagging attack.
In a separate work (in %
preparation) we present a modified photonic
classical Alice protocol 
that prevents many other tagging attacks, including the one 
suggested in~\cite{comment-on-140501} as an attack against QKD with classical Bob 
(\cite{boyer:140501}); 
see also~\cite{reply-140501}.

\fsubsection{Eve's attack on the (photonic) classical Alice protocol} %
Eve performs her attack in both directions using  
a fixed probe space $\mathscr{H}_E$; 
from Bob to Alice, Eve applies $U$; from Alice to Bob, Eve applies $V$.
The attack from Bob to Alice produces a state of the form
$
\sum \ket{E_{n_{\bito}n_{\bitz}}}\ket{n_{\bito},n_{\bitz}}_{AB} 
 \in \mathscr{H}_E\otimes\mathscr{H}_{AB}
$
where 
$\mathscr{H}_{AB} = \mathscr{F}$.
With Alice's probe attached we obtain
\begin{equation}
\ket{\Psi} = \sum \ket{E_{n_{\bito}n_{\bitz}}}\ket{0,0}_A\ket{n_{\bito},n_{\bitz}}  
\ ,
\end{equation}
in 
$\mathscr{H}_E\otimes \mathscr{H}_A \otimes \mathscr{H}_{AB}$.
In particular, if Eve does nothing then 
$\ket{E_{n_{\bito}n_{\bitz}}} \equiv \ket{E}$ independently of $n_{\bito}$ and $n_{\bitz}$, 
and the state in Alice+Eve's hands, prior to Alice's operation, is 
$\ket{E}\ket{0,0}_A\ket{\splus}_{AB}\ $.

Going back to the general case, if Alice applies $U_{\CTRL}$, 
then the state in Eve+Alice hands (after Alice's CTRL action) is still $\ket{\Psi}$.
However, if Alice applies $U_{\SIFT}$, 
the resulting global state in Eve+Alice's hands is
\[
\sum
\ket{E_{n_{\bito}n_{\bitz}}}\ket{\hat{n}_{\bito},\hat{n}_{\bitz}}_A\ket{n_{\bito},n_{\bitz}}_{AB}
\ ;
\]
after Alice has measured 
her probe she gets some output ($\{00,01,10,11\}$), 
and some complicated (non normalized) residual state (sent then back to Bob) that 
we soon analyze.

Eve now attacks that residual state on the way back from Alice to Bob 
using the 
unitary $V$ acting on both  
her probe and the state sent by Alice to Bob (see below).
Eve then sends Bob his part of the resulting state.

\fsubsection{A proof of robustness} %
Alice's measuring abilities put a constraint on the state $\ket{\Psi}$ for Eve not to be detectable: 
Alice's probability of measuring $\ket{11}_A$ according to that model must be
zero, or else Eve can be noticed.
It is thus required that $\ket{E_{n_{\bito}n_{\bitz}}} = 0$ for $n_{\bito} \times n_{\bitz} \neq 0$. 
Therefore, Eve+Alice's state when Alice applies $U_{\SIFT}$ 
must take the form
\begin{align*}
\sum_{n_{\bitz}\geq 1}& \ket{E_{0n_{\bitz}}}\ket{0,1}_A\ket{0,n_{\bitz}}+\sum_{n_{\bito}\geq 1}\ket{E_{n_{\bito}0}}\ket{1,0}_A\ket{n_{\bito},0} 
\\
&+\ket{E_{00}}\ket{00}_A\ket{0,0}
\ .
\end{align*}

Once Alice has performed her measurements and sent $\ket{i,j}_{AB}$ back to Bob via
Eve, the resulting global state (fully in Eve's hands)  is
\begin{center}
{\renewcommand{\arraystretch}{1.4}
\begin{tabular}{|c||l|}\hline
\textit{Measurement} & \textit{Residual state (in Eve's hands)}\\ \hline\hline
$00$ & $\ket{\psi_{00}}= \ket{E_{00}}\ket{0,0}_{AB}$ \\ \hline
$01$ & $\ket{\psi_{01}}=\sum_{n_{\bitz}\geq 1}
\ket{E_{0n_{\bitz}}}\ket{0,n_{\bitz}}_{AB}$ \\  \hline
$10$ & $\ket{\psi_{10}}=\sum_{n_{\bito}\geq 1}
\ket{E_{n_{\bito}0}}\ket{n_{\bito},0}_{AB}$ \\  \hline
CTRL & $\ket{\psi} = \ket{\psi_{00}} + \ket{\psi_{01}} +\ket{\psi_{10}}$ \\ \hline
\end{tabular}
}
\end{center}
where $\ket{\psi_{ij}}$ are not normalized, and where the $\ket{E_{ij}}$ were chosen by Eve. 
Eve now applies a unitary $V$ on $\mathscr{H}_E\otimes\mathscr{H}_{AB}$ 
and then
sends Bob his part of the resulting state.

Recall that Eve attacks now using the 
unitary $V$ acting on the residual state in $\mathscr{H}_E\otimes
\mathscr{H}_{AB}$, 
and then she sends Bob his part of the resulting state.
Bob's measuring abilities put more constraints on the state $\ket{\psi}$ 
for Eve not to be detectable. 
In case the SIFT bit is used for TEST,  
Bob's probability of measuring $11$ must be
zero, no matter what Alice measured. Furthermore, 
for Eve to stay undetectable, if Alice measured $\ket{0,0}$ 
(namely, the outcome $00$) in the SIFT mode, then Bob should have 
a probability zero of measuring $01$ or $10$, 
thus, a probability zero of receiving the states $\ket{1,0}$ or $\ket{0,1}$.
Similarly if Alice measured $10$ ($01$), 
then Bob should have a probability zero of measuring $01$ ($10$);
he could however get a loss, $00$.
The resulting (non normalized) Eve+Bob residual states thus take the form
$\ket{\psi_{00}'} = V \ket{\psi_{00}}  =
 \ket{H_{00}}\ket{0,0}_{AB} $ when a loss arrives, and  
\begin{align}
\ket{\psi_{01}'} &= V\ket{\psi_{01}} = 
\sum_{n_{\bitz}\geq 1}
\ket{F_{0n_{\bitz}}}\ket{0,n_{\bitz}}_{AB} 
+ \ket{H_{01}}\ket{0,0}_{AB}  \nonumber \\
\ket{\psi_{10}'} &= V\ket{\psi_{10}} = 
 \sum_{n_{\bito}\geq
1}\ket{F_{n_{\bito}0}}\ket{n_{\bito},0}_{AB} 
+ \ket{H_{10}}\ket{0,0}_{AB}  \label{final10}
\end{align}
otherwise;
$V$ being linear, the (normalized) residual state if Alice applied CTRL is 
$\ket{\psi'} \equiv V\ket{\psi} =
\ket{\psi_{00}'}+\ket{\psi_{01}'}+\ket{\psi_{10}'}$.

In order to check CTRL bits, Bob measures $\ket{\psi'}$
in the $x$ basis and checks if he gets at least one photon in any illicit state
such as $\ket{\sminus}$; more precisely,
he 
measures $\ket{\psi'}$ in the Fock basis $\ket{n_{\sminus},n_{\splus}}_x$ corresponding to the $x$ basis 
of single photon states,
and aborts if he gets $n_{\sminus}>0$ (if the detector for $\ket{\sminus}$ photons clicks).
To avoid that, Eve must make sure that the overlap between Eve-Bob's state 
$\ket{\psi'}$ and each state of the form 
$\ket{n_{\sminus},n_{\splus}}_x$ with 
$n_{\sminus}>0$ is zero.
This results with another limitation on Eve's attack.
We clarify in the appendix, sect. \ref{secfock} 
the expansion of the $x$-basis Fock states $\ket{n_{\sminus},n_{\splus}}_x$ 
using the $z$-basis Fock
states $\ket{n_{\bito},n_{\bitz}}$ and prove 
the following:

\begin{lemma}[Lemma]
If Bob has 
a {\it zero} 
probability of measuring any state $\ket{n_{\sminus},n_{\splus}}_x$ with $n_{\sminus} > 0$,  
then 
 $\ket{F_{01}}  =  \ket{F_{10}}$, and  $\ket{F_{0n}} = 
\ket{F_{n0}} = 0$ for $n>1$.
\end{lemma}

Letting $\ket{F} = \ket{F_{01}}=\ket{F_{10}}$, Eve+Bob's final residual states given 
by \eqref{final10},  
if Alice measured $01$ and $10$, are reduced to, strikingly, 
exactly the same states given (for the simpler case)
by \eqref{final-table} (respectively).
As before, if Bob measures in the $z$ basis and gets a SIFT bit, 
Eve's final state $\ket{F}$ is the same whether Bob measured $0$ or $1$ and she thus can get no information
on either Alice's measurement or Bob's result: the protocol is completely robust.

\fsection{Conclusions}%
From the above analysis we conclude that 
Bob must in the end, on CTRL bits, get
either a loss or exactly the state $\ket{\splus}$, 
which he thinks he sent. 
This does not mean that Eve's attack is trivial
(namely, she must send $\ket{\splus}$ to Alice, and do nothing on the 
way back). As the 
simplest non-trivial attack,
Eve could prepare the state $\ket{E}[\ket{0,2}+\ket{2,0}]/\sqrt{2}$,
and apply the transformation 
$V[\ket{E}\ket{0,2}] = \ket{E}\ket{0,1}; 
V[\ket{E}\ket{2,0}] = \ket{E}\ket{1,0}$
on the way back, 
without being noticed, but also, without gaining any information,
as we proved here. 

\fsection{Discussion}%
We presented  here 
a proof of robustness for two protocols 
in which Alice is classical, one that takes photon losses into 
account, and a more relevant one that also deals with multi-photon pulses.
The optimistic conclusion of robustness here is, 
unfortunately, not the end
of the story, and further research is required: 
First, we dealt in this paper only with the generalization of the qubits
of ``QKD with classical Alice'' into two modes, and we left the case of more 
modes open. 
Second, here we let 
almighty Eve prepare the state;
unfortunately, Bob is not as capable as Eve, and in reality, he
is the one preparing the state, not Eve;
Bob, who tries to generate the state $\ket{0,1}_x$,
may be unable to avoid (sometimes)
sending the state $\ket{0,2}_x$ which will often cause
a 11 reading 
in the computation basis, and destroy the full robustness.
Still, there is evidence (yet, no proof) 
that the classical Alice protocol is more robust than BB84.
See appendix, sect. \ref{sectD} 
 where we also propose
three ways to improve the partial robustness (or the security) of our protocol, and of BB84.

\bigskip
\begin{center}
\large\textbf{Appendix}
\end{center}

\newcommand\tr[0]{\operatorname{tr}}
\newcommand\had[0]{\mathrm{H}}
\newcommand\lossrate[0]{\mathrm{lossrate}}
\newtheorem{lemnot}{Lemma}

\def\appendixname{Section}

\appendix
\section{On the Robustness of the Ben92 Scheme}\label{sectA}
Let Alice generate a string of qubits by choosing randomly (with equal probability) one of the two distinct and non orthogonal states
$\ket{u_\bitz}$ and $\ket{u_1}$ and sending it to Bob via a quantum channel. 
Bob measures the incoming qubits at random either in the orthonormal basis 
$\{\ket{u_\bitz}, \ket{u'_\bitz}\}$ or
in the orthonormal basis $\{\ket{u_\bito}, \ket{u'_\bito}\}$. 
If he measures $\ket{u'_\bito}$, he can be sure he was sent $\ket{u_\bitz}$
because he can't have been sent $\ket{u_\bito}$. 
Similarly, if he measures $\ket{u'_\bitz}$, he is certain he was sent $\ket{u_\bito}$.
Note that if Alice sent $\ket{u_\bitz}$ and Bob chose (with a probability half)
to measure in the $\{\ket{u_\bitz}, \ket{u'_\bitz}\}$ basis, he will surely get 
$\ket{u_\bitz}$, hence an
inconclusive result. 
With that procedure, the probability of a conclusive (un-ambiguous)
measurement~\cite{endnote9}
is 
\begin{equation}\label{pconclusivesimple}
p_{\text{conclusive}} = (1/2)[1-|\braket{u_\bitz}{u_\bito}|^2] \ .
\end{equation}

The procedure is robust against eavesdropping because if he was sent $\ket{u_\bitz}$, and he measured
in the $\{\ket{u_\bitz}, \ket{u'_\bitz}\}$ basis, and he did not get
$\ket{u_\bitz}$, then he
knows the incoming state has been tampered with; 
similarly, if he was sent $\ket{u_\bito}$ and measured $\ket{u'_\bito}$.
For a security analysis one must allow some small probability of noise, hence of
errors and/or losses.

If high loss-rate cannot be avoided, which is a typical case in QKD,
the Ben92 scheme as described here becomes totally non-robust; Bennett \cite{Ben92} was, of
course, aware of this, hence designed his protocol differently. 
In case high losses must be tolerated, such that 
$(\lossrate) \ge 1-p_{\text{conclusive}} 
= (1/2)[1+|\braket{u_\bitz}{u_\bito}|^2]$, 
an eavesdropper can simply catch all the qubits coming out of Alice's hands, 
measure according to Bob's procedure, and send Bob (via a lossless channel) the proper state only when the measurement was
conclusive, else send nothing. 
This attack was called ``conclusive attack'' in the past, but later on the term
``un-ambiguous state discrimination'' became more popular then the term
``conclusive''.

It is interesting to note (although not vital for the current paper,
hence we skip the details here)
that Eve can even do better than Bob, since she is more
powerful, using generalized measurements (POVMs) as described
in~\cite{Ivanovic87,Peres88,Dieks88}; see also~\cite{EHPP94,HIGM95,Yuen}.
Thus it can be shown that even if the lossrate is 
below  
$(1/2)[1+|\braket{u_\bitz}{u_\bito}|^2]$, yet as long as it is above 
$|\braket{u_\bitz}{u_\bito}|$,  
the protocol is still totally non-robust.

\section{On the Robustness of the BB84 Scheme}\label{sectpns}
This section should be read after the paragraphs
concerning the Fock notations for pulses with indistinguishable photons in the main article.

The photon number splitting attack was introduced in \cite{BLMS00}.  
Here is a short description in
the notations and the framework of the current article.

\subsection{Nondemolition-splitting of two photon pulses}\label{splitphoton}

We assume that Eve has an initial probe $\ket{0,0}^{\mathrm{E}}$ in the Fock space. Photonic states $\ket{n_\bito,n_\bitz}$ sent from Alice to Bob are attacked with $U\ket{0,0}^{\mathrm{E}}\ket{n_\bito,n_\bitz} = \ket{0,0}^{\mathrm{E}}\ket{n_\bito,n_\bitz}$ if $n_\bitz+n_\bito \neq 2$ and
\begin{align}
U\ket{0,0}^{\mathrm{E}}\ket{0,2} &= \ket{0,1}^{\mathrm{E}}\ket{0,1} \label{ket02}\\
U\ket{0,0}^{\mathrm{E}}\ket{2,0} &= \ket{1,0}^{\mathrm{E}}\ket{1,0} \label{ket20}\\
U\ket{0,0}^{\mathrm{E}}\ket{1,1} &= \frac{1}{\sqrt{2}}\big[\ket{1,0}^{\mathrm{E}}\ket{0,1} + \ket{0,1}^{\mathrm{E}}\ket{1,0}\big].\label{ket11}
\end{align}
The first two equations mean that if two photons in the $\bket{0}$ state, 
i.e. $\ket{0,2}$, or in 
the $\bket{1}$ state, i.e. $\ket{2,0}$, are sent, Eve keeps one.
The third equation is required for the same to hold when the 
two photons are in the $\ket{\splus}$ state, 
i.e. $\ket{0,2}_x$, or in 
the $\ket{\sminus}$ state, i.e. $\ket{2,0}_x$. 
Thus, $U$ describes a nondemolition-splitting of two photons if those are
prepared in the standard or Hadamard bases.

Let us now demonstrate the effect of $U$ in the $x$-basis.
For one photon pulses 
\begin{align}
\ket{0,1}_x &= \bket{+} = \frac{1}{\sqrt{2}}[\bket{0}+\bket{1}] = \frac{1}{\sqrt{2}}[\ket{0,1} + \ket{1,0}] \label{ket01} \\
\ket{1,0}_x &= \bket{-} =  \frac{1}{\sqrt{2}}[\bket{0}-\bket{1}] = \frac{1}{\sqrt{2}}[\ket{0,1}-\ket{1,0}] \label{ket10}
\ .
\end{align}
To express $\ket{2,0}_x$ as a superposition of the states $\ket{0,2}$, $\ket{1,1}$ and $\ket{2,0}$,
we need to discuss how to deal with indistinguishable particles; more details can be found in Sect \ref{secfock}.
The state $\ket{2,0}$ corresponds to two indistinguishable photons in the state $\ket{\sminus}$ which is a pulse
containing a state similar to $\ket{-\kern-0.1em -}$, but with no importance to the order of the
two single-photon states. Because the photons are indistinguishable, we can write
\[
\ket{-\kern-0.1em -} = \frac{1}{2}\big[\bket{0}-\bket{1}\big]\otimes\big[\bket{0}-\bket{1}\big]= \frac{1}{2} \big[\bket{00}-\bket{01}-\bket{10} + \bket{11}\big]
\]
with indistinguishable particles within each term.
The state $\bket{00}$ means two identical photons in the $\bket{0}$ state, and corresponds to $\ket{0,2}$. Similarly
the state $\bket{11}$ corresponds to two identical photons in the $\bket{1}$ state, i.e. to $\ket{2,0}$. 
Remains, after normalizing, the (already symmetric) state $\frac{\bket{01} + \bket{10}}{\sqrt{2}}$
and that state corresponds  to one photon in state $\bket{0}$ and one photon in state $\bket{1}$ permuted in all possible
ways, and then normalized, which corresponds to state $\ket{1,1}$. 
Since $\bket{--}$ for two photons equals 
$ \frac{1}{2}\big[\bket{00} - \sqrt{2}\frac{\bket{01}+\bket{10}}{\sqrt{2}} + \bket{11}\big]$ and, 
similarly, $\bket{++}$ is equal to
$ \frac{1}{2}\big[\bket{00} + \sqrt{2}\frac{\bket{01}+\bket{10}}{\sqrt{2}} +
\bket{11}\big]$, we conclude that 
\begin{align}
\ket{0,2}_x &= \frac{1}{2}\left[ \ket{0,2}+\sqrt{2}\ket{1,1} + \ket{2,0}\right] \label{twoplusses}\\
\ket{2,0}_x &= \frac{1}{2}\left[ \ket{0,2} - \sqrt{2}\ket{1,1} + \ket{2,0}\right] \label{twominusses}
\ .
\end{align}
These two identities are obtained with an alternative method (raising operators) in Sect \ref{subraising}.
From the definition of $U$ above, we can now derive the equalities
\begin{align}
U\ket{0,0}^{\mathrm{E}}\ket{0,2}_x &= \ket{0,1}^{\mathrm{E}}_x \ket{0,1}_x \label{keep0x}\\
U\ket{0,0}^{\mathrm{E}}\ket{2,0}_x &=   \ket{1,0}^{\mathrm{E}}_x \ket{1,0}_x\label{keep1x}
\ .
\end{align}
Here is the derivation of \eqref{keep1x}:
\begin{align*}
U\ket{0,0}^{\mathrm{E}}\ket{2,0}_x &= U\ket{0,0}^{\mathrm{E}}\left(\frac{\ket{0,2} - \sqrt{2}\ket{1,1}+ \ket{2,0}}{2}\right)&\text{ by \eqref{twominusses}}\\
&= \frac{\ket{0,1}^{\mathrm{E}}\ket{0,1}}{2} & \text{by \eqref{ket02}}\\
 &- \frac{\ket{1,0}^{\mathrm{E}}\ket{0,1}}{2} - \frac{\ket{0,1}^{\mathrm{E}}\ket{1,0}}{2} & \text{by \eqref{ket11}}\\
 &+\frac{\ket{1,0}^{\mathrm{E}}\ket{1,0}}{2}&\text{by \eqref{ket20}}\\
 &= \left(\frac{\ket{0,1}^{\mathrm{E}}-\ket{1,0}^{\mathrm{E}}}{\sqrt{2}}\right)\left(\frac{\ket{0,1}-\ket{1,0}}{\sqrt{2}}\right)\\
 &= \ket{1,0}^{\mathrm{E}}_x \ket{1,0}_x&\text{by \eqref{ket10}}
\end{align*}
The derivation of \eqref{keep0x} is identical, with `$+$' everywhere instead of `$-$'. Equations \eqref{ket02} and \eqref{ket20}
together with \eqref{keep0x} and \eqref{keep1x}
mean that Eve's attack is unnoticed both in the $z$ and the $x$ basis on two
identical photons: Bob receives a single photon, ``undisturbed'', and 
Eve gets full information when the basis is published.

\subsection{The PNS attack on the BB84 protocol}

In the BB84 protocol~\cite{BB84}, photons go from Alice to Bob.
For each choice of basis $b$ ($z$ or $x$) and each bit chosen randomly, it is assumed that Alice sends
$\ket{0,0}$ with probability $p_0$, $\ket{0,1}_b$ or $\ket{1,0}_b$ with probability $p_1$, and
$\ket{0,2}_b$ or $\ket{2,0}_b$ with probability $p_2$, where $\ket{n_\bito,n_\bitz}_z \equiv \ket{n_\bito,n_\bitz}$. 
We may assume that $p_0+p_1+p_2 = 1$.
We also assume $p_2\ll p_1$, and a loss rate close to $100\%$, i.e.
$F=1-(\lossrate) \ll 1$.

In $N$ trials, Bob expects $Fp_1 N$ single photon pulses from the expected $p_1 N$ single photon pulses coming from Alice.
From the expected $p_2 N$ two-photon pulses coming from Alice, the chances that the two photons will be lost
in the channel are  $(\lossrate)^2 = (1-F)^2$.
The chances that at least one photon reaches Bob are thus  $1 - (1-F)^2$ and so 
Bob expects a total of
\[
X = \left(F p_1 + [1 - (1-F)^2] p_2\right) N
\] non empty pulses.
In the PNS attack Eve makes sure Bob gets the number of pulses he is expecting.
If
\[
p_2 N \geq X = \left( F p_1 + [1 - (1-F)^2] p_2\right) N
\]
i.e. 
\[
\frac{p_2}{p_1} > \frac{F}{(1-F)^2}
\]
then the number of two photon pulses emitted by Alice, namely $p_2N$, is larger than the number of pulses Bob
 is expecting.
With the non lossy channel, 
Eve can simply select $X$ two photon pulses from those $p_2N$ pulses sent by Alice (she is able to count photons),
and attack them with the two photon pulse attack (Eq. \ref{ket02}-\ref{keep1x}), keep one photon and send Bob the other.
She thus sends $X$ single photon pulses to Bob;
 there is no way for Bob 
to check for eavesdropping; he receives exactly the number of pulses he is expecting, and as they should have 
been generated in the first place in the ideal qubit protocol.
The BB84 protocol is thus completely non robust as soon as $p_2/p_1 \geq F/(1-F)^2$ which, to a first
order approximation~\cite{endnote10},
holds when the rate of two-photon pulses amongst the non empty pulses is larger than $F$.

\section{Changing Basis in the Fock Space and the Proof of the Lemma}\label{secfock}

The quantum states of the photons manipulated by Alice and Bob can be described as states in the Fock space $\mathscr{F}$ whose Hilbert basis is given by the Fock states
$\ket{n_\bito,n_\bitz}$ where $n_\bito$ is the number of indistinguishable photons in the $\bket{1}$ state and $n_\bitz$ the number of 
indistinguishable photons in the $\bket{0}$ state. Alice and Bob however also use the
Hadamard basis $\ket{+}, \ket{-}$ and we also need to use the Fock states $\ket{n_-,n_+}_x$ as a basis for $\mathscr{F}$,
where $\ket{n_-,n_+}_x$ corresponds to $n_-$ indistinguishable  photons in the $\ket{-}$ state
and $n_+$ indistinguishable  photons in the $\ket{+}$ state. The states $\ket{n_\bitz,n_\bito}$ and 
$\ket{n_-,n_+}_x$ belong to the same space of states. How are they related? 
A pulse $\ket{n_-,n_+}_x$ with $n_-$ indistinguishable photons in the $\ket{-}$ state and 
$n_+$ in the $\ket{+}$ state can always be expressed as a superposition of pulses $\ket{n_\bito,n_\bitz}$
with $n_\bito$ photons in the $\bket{1}$ state and $n_\bitz$ photons in the $\bket{0}$ state such that
$n_+ + n_- = n_\bitz+n_\bito$. In this paper, we need to know the coefficients of that superposition when
 either $n_+$ or $n_-$ is zero.

\subsection{The symmetric state method}
We already presented formulas for $\ket{0,2}_x$ and $\ket{2,0}_x$ in Sect \ref{sectpns}, namely
Eq. \eqref{twoplusses} and
Eq \eqref{twominusses}. The very same reasoning can be 
applied with three indistinguishable photons in the $\bket{-}$ state.
Expanding $\ket{-}^{\otimes 3}$, which corresponds to three photons in the $\ket{-}$ state, thus $\ket{3,0}_x$,
gives
\[
\frac{\bket{000} - \bket{001} - \bket{010} + \bket{011} -
\bket{100} + \bket{101} + \bket{110} - \bket{111}}{\sqrt{8}}
\]
and using the following normalized states as a representation for the Fock states
\begin{align*}
\ket{0,3} &= \bket{000} &\text{all photons in the $\bket{0}$ state}\\
\ket{1,2} &= \frac{1}{\sqrt{3}}[\bket{100}+\bket{010}+\bket{001}] \\
\ket{2,1} &= \frac{1}{\sqrt{3}}[\bket{110} + \bket{101} + \bket{011}] \\
\ket{3,0} &= \bket{111} &\text{all photons in the $\bket{1}$ state}
\end{align*}
we obtain
\[
\ket{3,0}_x = \frac{1}{\sqrt{8}}\left[\ket{0,3} -\sqrt{3}\ket{1,2} + \sqrt{3}\ket{2,1} - \ket{3,0}\right]
\]
Of course, $\ket{0,3}_x$ gives the same expansion, but with $+$ everywhere.
This reasoning generalizes to $\ket{0,n}_x$ and $\ket{n,0}_x$ provided we accept
that $n-k$ photons in the $\bket{0}$ state and $k$ photons in the $\bket{1}$ state
corresponds to the equal superposition of all $n$ qubit basis states $\ket{j}$ that have
$k$ bits equal to $\bito$ and thus $n-k$ bits equal to $\bitz$, i.e. Hamming weight $|j| = k$. Notice that the number
of $n$-bit strings with Hamming weight $k$ is  $\binom{n}{k}$ and the normalizing factor is thus
$\binom{n}{k}^{-1/2}$; the symmetric state representation is
\begin{align}
\ket{k,n-k}\ =\ \binom{n}{k}^{-1/2} \sum_{\stackrel{\scriptstyle j\in\{\bitz,\bito\}^n}{|j| = k}} \ket{j}\label{ssrep}
 \end{align}
Using the well known (within the quantum information community \cite[p 35]{NielsenChuang2000}) formula 
\begin{equation}\label{hadamardn}
\had^{\otimes n}\ket{i} = \frac{1}{\sqrt{2}^n}\sum_{j\in \{\bitz,\bito\}^n} (-1)^{i\cdot j}\ket{j},
\end{equation}
 where $\had$ is the Hada\-mard transform i.e.
$\had\bket{0} = \ket{+}$ and $\had\bket{1} = \ket{-}$, we deduce the general
formula
\begin{align}
\ket{0,n}_x &= \frac{1}{{\sqrt{2}}^n} \sum_{k=0}^n \binom{n}{k}^{1/2} \ket{k,n-k} \label{nzeros}
\end{align}
as follows:
\begin{align*}
\ket{0,n}_x &= \frac{1}{\sqrt{2}^n} \sum_{j\in\{\bitz,\bito\}^n} \ket{j} &\text{by \eqref{hadamardn}}\\
&=  \frac{1}{\sqrt{2}^n} \sum_{k=0}^n\ \sum_{|j|=k}\ \ket{j}&\text{group by number of ones}\\
&=  \frac{1}{\sqrt{2}^n} \sum_{k=0}^n \binom{n}{k}^{1/2} \ket{k,n-k} &\text{by \eqref{ssrep}}
\end{align*}
Similarly, $\ket{n,0}_x$ corresponds to $\ket{-}^{\otimes n} = \left(\had\bket{1}\right)^{\otimes n} =
\had^{\otimes n}\ket{i}$  for $i=\bket{1\ldots 1}$ ($n$ ``$\bito$'' bits); then $(-1)^{i\cdot j} = (-1)^{|j|} = (-1)^k$ in
\eqref{hadamardn} and 
\begin{align}
\ket{n,0}_x&= \frac{1}{{\sqrt{2}}^n} \sum_{k=0}^n (-1)^k \binom{n}{k}^{1/2}\ket{k,n-k}. \label{nones}
\end{align}
Since it is also true that $\had\ket{+} = \bket{0}$ and $\had\ket{-} = \bket{1}$, formulas
\eqref{nzeros} and \eqref{nones} hold if we move the index $x$ from the left to the
right to express $\ket{0,n}$
and $\ket{n,0}$ in terms of the $\ket{k,n-k}_x$.

\subsection{The raising operator method}\label{subraising}
That method should be more congenial to anyone having some knowledge of
quantum field theory or quantum optics. It can be shown that to  the standard basis $\bket{0}$ and
$\bket{1}$ corresponds a set of two commuting operators that we will denote $a_\bitz^\dagger$ and
$a_\bito^\dagger$ such that
\begin{equation}\label{raisingstandard}
{a_\bito^\dagger}^{n_\bito} {a_\bitz^\dagger}^{n_\bitz} \ket{0,0}=\sqrt{n_\bito!n_\bitz!}\ \ket{n_\bito,n_\bitz}
\end{equation}
Similarly, to the basis $\ket{+},\ket{-}$ corresponds the set of commuting operators
$a_+^\dagger$ and $a_-^\dagger$, and they are such that
\begin{equation}\label{raisinghadamard}
{a_-^\dagger}^{n_-} {a_+^\dagger}^{n_+}\ket{0,0} =\sqrt{n_-!n_+!}\ \ket{n_-,n_+}_x
\end{equation}
Moreover, relating the operators $a_+^\dagger$ and $a_-^\dagger$ to
$a_\bitz^\dagger$ and $a_\bito^\dagger$ is quite straightforward. From
$\ket{+} = \frac{1}{\sqrt{2}}[\bket{0}+\bket{1}]$ and $\ket{-} = \frac{1}{\sqrt{2}}[\bket{0}-\bket{1}]$
we are allowed to deduce
\begin{align}
a_+^\dagger = \frac{1}{\sqrt{2}}\left[a_\bitz^\dagger + a_\bito^\dagger\right],\quad
a_-^\dagger = \frac{1}{\sqrt{2}}\left[a_\bitz^\dagger - a_\bito^\dagger\right] .\label{chbas}
\end{align}
All calculations are then direct without any
intermediate symmetric state representation. Here they are for 
$\ket{2,0}_x$ (using the fact
that ${a_+^\dagger}^0$ is the identity and $0! = 1$):
\begin{align*}
\ket{2,0}_x &= \frac{1}{\sqrt{2!}} {a_-^\dagger}^2\ket{0,0}
&\text{by \eqref{raisinghadamard}}\\
&= \frac{1}{\sqrt{2}^2}\frac{1}{\sqrt{2!}} [a_\bitz^\dagger - a_\bito^\dagger]^2\ket{0,0}
&\text{by \eqref{chbas}}\\
&= \frac{1}{\sqrt{2}^2}\frac{1}{\sqrt{2!}} \left[{a_\bitz^\dagger}^2 - 2 {a_\bitz^\dagger}{a_\bito^\dagger} + {a_\bito^\dagger}^2\right]\ket{0,0}\\
&=  \frac{1}{\sqrt{2}^2}\frac{1}{\sqrt{2!}}\left[
\sqrt{2!}\ket{0,2} -2\sqrt{1!1!}\ket{1,1} + \sqrt{2!}\ket{2,0} \right]
&\text{by \eqref{raisingstandard}}\\
&= \frac{1}{2}\left[\ket{0,2} -\sqrt{2}\ket{1,1} + \ket{2,0}\right]
\end{align*}
which coincides with \eqref{twominusses}. Newton's binomial expansion may be applied because the operators commute. 
To get \eqref{nzeros} for $\ket{0,n}_x$ and \eqref{nones} for $\ket{n,0}_x$ one needs only
follow the same reasoning as
above with $n$ instead of $2$:

\begin{align*}
\ket{n,0}_x &=\frac{1}{\sqrt{n!}} {a_-^\dagger}^n\ket{0,0}\\ %
 &=  \frac{1}{\sqrt{n!}}\left[\frac{a_\bitz^\dagger - a_\bito^\dagger}{\sqrt{2}}\right]^n\kern-0.5em\ket{0,0}\\ %
 &= \frac{1}{\sqrt{2}^n}\frac{1}{\sqrt{n!}}\sum_{k=0}^n \frac{n!}{(n-k)!k!} (-1)^k{a_\bitz^\dagger}^{n-k}{a_\bito^\dagger}^k\ket{0,0}\\ %
 &=  \frac{1}{\sqrt{2}^n}\frac{1}{\sqrt{n!}}\sum_{k=0}^n \frac{n!}{(n-k)!k!} (-1)^k\sqrt{(n-k)!k!}\ \ket{k,n-k}\\ %
 &=  \frac{1}{\sqrt{2}^n}\sum_{k=0}^n (-1)^k \binom{n}{k}^{1/2}  \ket{k,n-k}.
\end{align*}
 
 \subsection{Two distinguished states}
 The amplitudes in Eq \eqref{nzeros} and Eq \eqref{nones} are exactly the same, but to a phase factor,
 and in both cases the distribution of the number of identical photons in state $\bket{1}$ if we repeat the measurement in the standard basis is the binomial
 $B(n,p=1/2)$, for which $f(k;n,p) = \binom{n}{k}p^k(1-p)^{n-k}$ with $p=1/2$. However,
there is a negative sign on the amplitude for odd values of $k$ in
 Eq \eqref{nones}. That means that those amplitudes cancel when $\ket{0,n}_x$ and $\ket{n,0}_x$ are added. That interference gives two
  distinguished states:
 \begin{align*}
 \ket{e^{(n)}}_x= \frac{\ket{0,n}_x+\ket{n,0}_x}{\sqrt{2}} &= \sum_{\stackrel{\scriptstyle k=0}{\scriptstyle k \equiv 0 \pmod 2}}^n  \left(\frac{\binom{n}{k}}{2^{n-1}}\right)^{1/2} \ket{k,n-k}\\ %
 \ket{o^{(n)}}_x=  \frac{\ket{0,n}_x-\ket{n,0}_x}{\sqrt{2}} &= \sum_{\stackrel{\scriptstyle k=0}{\scriptstyle k \equiv 1 \pmod 2}}^n   \left(\frac{\binom{n}{k}}{2^{n-1}}\right)^{1/2}  \ket{k,n-k} %
 \end{align*}
 where ``$k \equiv 0 \pmod 2$'' means ``$k$ even'', the sum being over all even values of $k$,  and ``$k \equiv 1 \pmod 2$'' means ``$k$ odd'', the sum
 being over all odd values of $k$.
 The states $\ket{e^{(n)}}_x$ and $\ket{o^{(n)}}_x$ are clearly orthogonal. Moreover, if we measure state $\ket{e^{(n)}}_x$ (resp. state $\ket{o^{(n)}}_x$) %
 in the standard basis, 
 the probability of getting $k$ identical photons in state $\bket{1}$ for $k$ even (resp. for $k$ odd) is now $2f(k; n,1/2)$, twice that given by the binomial distribution, and it is $0$ otherwise:
 we never ever get and odd number of $\bket{1}$ photons with state 
$\ket{e^{(n)}}_x$ %
and never ever get an even number of $\bket{1}$ photons
 with state 
$\ket{o^{(n)}}_x$.%
 
 Similarly, if $\ket{e^{(n)}} = \big(\ket{0,n}+\ket{n,0}\big)/\sqrt{2}$ and $\ket{o^{(n)}} = \big(\ket{0,n}-\ket{n,0}\big)/\sqrt{2}$ then
 \begin{align}
 \ket{e^{(n)}} &=  \sum_{\stackrel{\scriptstyle k=0}{\scriptstyle k \equiv 0 \pmod 2}}^n  \left(\frac{\binom{n}{k}}{2^{n-1}}\right)^{1/2} \ket{k,n-k}_x \label{evenminus}\\
 \ket{o^{(n)}} &=  \sum_{\stackrel{\scriptstyle k=0}{\scriptstyle k \equiv 1 \pmod 2}}^n  \left(\frac{\binom{n}{k}}{2^{n-1}}\right)^{1/2} \ket{k,n-k}_x \label{oddminus}
  \end{align}
  \subsection{The concluding lemma}
 To prove the robustness of the photonic classical Alice with losses and multi photon pulses, we relied on the 
 following result:
 \begin{lemnot}[Lemma]
Given the bipartite state $\ket{\psi'} = \ket{H}\ket{0,0} + 
\sum_{n_{\bitz}\geq 1} \ket{F_{0n_{\bitz}}}\ket{0,n_{\bitz}} 
+ \sum_{n_{\bito}\geq 1}\ket{F_{n_{\bito}0}}\ket{n_{\bito},0}$ in $\mathscr{H}_E\otimes \mathscr{F}$,
if there is a zero probability of measuring any basis states  $\ket{n_-,n_+}_x$ of $\mathscr{F}$ such that $n_- > 0$, then
 $\ket{F_{01}}  =  \ket{F_{10}}$, and  $\ket{F_{0n}} = 
\ket{F_{n0}} = 0$ for $n>1$.
\end{lemnot}
\begin{proof}
The overlap of $\ket{n_-,n_+}_x$ with $\ket{0,0}$ is $0$; so is its overlap of $\ket{n_-,n_+}_x$ with any
$\ket{n_\bito,0}$ or $\ket{0,n_\bitz}$ for
which $n_\bitz\neq n_++n_-$ or $n_\bito \neq n_++n_-$. We thus need consider only cases where $n  = n_\bitz = n_\bito  = n_++n_-$
and thus the overlap of $\ket{n_-, n - n_-}_x$ with the state 
$\ket{F_{0n}}\ket{0,n} + \ket{F_{n0}}\ket{n,0}$. 
 A simple calculation shows that
\[
 \ket{F_{0n}}\ket{0,n}+ \ket{F_{n0}}\ket{n,0}
 \]
  is equal to
\begin{equation}\label{difwithn}
  \left[\frac{\ket{F_{0n}}+\ket{F_{n0}}}{\sqrt{2}}\right]\ket{e^{(n)}}+  \left[\frac{\ket{F_{0n}}-\ket{F_{n0}}}{\sqrt{2}}\right]\ket{o^{(n)}}
\end{equation}
For $n=1$ the probability of measuring $\ket{1,0}_x$ must be $0$. Since $\braket{e^{(1)}}{1,0}_x = 0$ (because $1$ is odd) and $\braket{o^{(1)}}{1,0}_x = 1$ 
the probability of measuring $\ket{1,0}_x$ is zero iff $\big[\ket{F_{01}} - \ket{F_{10}}\big]/\sqrt{2} = 0$ i.e. $\ket{F_{01}} = \ket{F_{10}}$.

For $n=2$ Eves must make sure the probability of measuring both $\ket{1,1}_x$ and $\ket{2,0}_x$ must be $0$.
From  $\braket{e^{(2)}}{1,1}_x = 0$ and $\braket{o^{(2)}}{1,1}_x = 1$, the probability of measuring $\ket{1,1}_x$ is zero iff
$\big[\ket{F_{02}}-\ket{F_{20}}\big]/\sqrt{2} = 0$, i.e. if $\ket{F_{02}} = \ket{F_{20}}$.
On the other hand $\braket{o^{(2)}}{2,0}_x = 0$ and $\braket{e^{(2)}}{2,0}_x = 1/\sqrt{2}$ and the probability of measuring $\ket{2,0}_x$ is zero
iff $\big[\ket{F_{02}}+\ket{F_{20}}\big]/2 = 0$ i.e. if $\ket{F_{02}} = - \ket{F_{20}}$. For both probabilities to be zero, it is necessary
and sufficient that $\ket{F_{02}} = \ket{F_{20}} = 0$.

For $n>2$, any odd $k$ is such that 
$\braket{e^{(n)}}{k,n-k}_x = 0$, and for the vector coefficient of $\ket{o^{(n)}}$ in \eqref{difwithn}
to be $0$,
$\ket{F_{n0}} = \ket{F_{0n}}$ is required; similarly, for any $k>0$ even, $\braket{o^{(n)}}{k,n-k}_x = 0$,
which implies $\ket{F_{0n}} = -\ket{F_{n0}}$, and thus
 $\ket{F_{0n}}=\ket{F_{n0}} = 0$.
\end{proof}

\section{Extended Discussion}\label{sectD}

\subsection{PNS attack on QKD with classical Alice when Bob may send two photon pulses}

As in Appendix B, when two-photon pulses are sometimes sent,
we can no longer get a proof of full robustness.
Still, we provide here some evidence that QKD with classical Alice is 
potentially more robust than BB84. Let us examine an extreme case.
In BB84 (see Sect. \ref{sectpns}), if the originator (Alice)
sends (without being aware of it) only two-photon pulses, Eve gets full
information without being noticed, and nothing in the tests performed by Alice
and Bob can reveal the deviation from the original protocol. 
In QKD with classical Alice, it is not so. If the originator (Bob) 
sends (without being aware of it) only two-photon pulses, Eve gets full
information via the nondemolition splitting as in section \ref{splitphoton},
yet now, Alice and Bob can easily notice the deviation from
the ideal protocol:
Alice will notice that on half of the SIFT bits both her detectors click.
  
Let us also consider the case in which Bob sometimes generates two-photon pulses,
and the loss rate is very large, much beyond that one considered in section \ref{sectpns},
so that $p_2 \gg p_1 F$. We have seen that the BB84 protocol is then totally
non-robust (Eve gets full information).
Is it also true for QKD with classical Alice?
If Eve blocks all single-photon pulses, again, half the non-empty 
SIFT bits will cause both
detectors to click, and Eve will be noticed.
To keep Alice having less than half such illegitimate detections, Eve can 
either let some single-photon pulses go to Alice, or Eve can send,
in addition to $\ket{0,2}_x$, states such as
$\ket{2,0}$ and $\ket{0,2}$, their superpositions and their mixtures.
States such as 
$\ket{2,0}$ and $\ket{0,2}$, or their mixtures, will cause errors in case Alice
applies CTRL and Bob measures in the $x$ basis.
States such as 
$[\ket{2,0}+\ket{0,2}]/\sqrt{2}$, which we already met more than once here,
will not cause errors in case Alice use CTRL conditioned on Eve applying the
proper transformation $V$ such that the state received by Bob is ``$\splus$''.
However, such states are as robust as the state ``$\splus$'' of the untouched
protocol, thus we may conclude (although we do not attempt to provide a full
proof here) that the further Eve's attack is from causing half the SIFT bits to
have double-detections, the more robust the protocol is.

\subsection{Three ways of strengthening our QKD with classical Alice protocol}

We have seen in the main paper a  
simple, yet non-trivial attack:
Eve could prepare the state $\ket{E}[\ket{0,2}+\ket{2,0}]/\sqrt{2}$,
send it to Alice, 
and apply the transformation
$V[\ket{E}\ket{0,2}] = \ket{E}\ket{0,1}; 
V[\ket{E}\ket{2,0}] = \ket{E}\ket{1,0}$
on the way back,
without being noticed, but also, without gaining any information,
as we proved.
Although Eve gains nothing from that attack, it is potentially disturbing.
It means that Eve can {\em totally} deviate from the protocol without being noticed,
and such a situation is not desired; it could have a strong impact on security
when noise is allowed, and/or when Bob sometimes sends $\ket{0,2}_x$.

We now present three ways of improving the protocol and potentially 
making it more secure (although a security analysis is beyond the scope of this
paper); two of these methods prevents Eve from applying the above-mentioned
attack.

\begin{itemize}
\item
Technology improvement: 
Replacing the detectors by counters that can (at least) distinguish a single
photon from more than one photon. Obviously, Alice's ability to distinguish in
SIFT mode a single photon from more than one photon prevents the above-mentioned
attack. Also, it allows Alice to obtain meaningful statistics in case Bob 
sometimes sends two photons in the state $\ket{0,2}_x$, 
as the case of Alice measuring $\ket{1,1}$ can now be compared to the cases of
measuring $\ket{2,0}$ and $\ket{0,2}$. 
\item 
Algorithmic improvement:
by adding more tests into the protocol we can improve its potential security.
So far we only discussed the case in which 
Alice applies SIFT and Bob measures in the $z$ basis, and the case in which Alice
applies CTRL and Bob measures in the $x$ basis. However, 
Alice and Bob can easily add two tests: Alice applies SIFT and Bob measures in the $x$
basis, and Alice applies CTRL and Bob measures in the $z$ basis; such a
modification could happen anyhow in real life QKD, because Bob is not currently using a
quantum memory; see end note 10 in the main paper. 
While these tests do not help against the above-mentioned attack they do help
having a better estimate of the states Bob generates: in case Bob sometimes  generates
$\ket{0,2}_x$, this can be noticed as a measurement of (1,1) in both
Bob's detectors when Alice employs CTRL and Bob measures in the $z$ basis, and
similarly (a detection in both Bob's detectors) when Alice employs SIFT 
yet does not get (1,1), and Bob measures in the $x$ basis.
\item
Protocol modification: Let us allow quantum Bob to add more states: 
We noted that in ``QKD with classical Bob'' the quantum originator sent not
only $\ket{\splus}$ but also other states such as $\ket{\bitz}$. Then the quantum
originator and the classical party performed their TEST on qubits going from the
quantum originator to the classical party.  In contrast, in QKD with classical
Alice,
one only defined the TEST on qubits going back from classical Alice to
(quantum) Bob.  We could allow our quantum Bob send also the states
$\ket{\bitz}$ and $\ket{\bito}$ and let him and classical Alice use those added
qubits only for an additional TEST, comparing bits when Bob generated these
states and Alice applied SIFT. Such a modification trivially prevents Eve from
applying the above-mentioned attack, since the protocol involves (on the way to
Alice) one of three
non-orthogonal states in each transmission, thus if Eve always sends the
above-mentioned state, she will be easily detected. 
\end{itemize}

Each of those modifications could only strengthen 
the protocol. Potentially they can also be combined together. 
The use of counters, and the use of tests in different bases could also be
helpful for improving BB84.


\begin{thebibliography}{1}
\bibitem{endnote9}{ A hint that one can do much better: if the two sent states are near-orthogonal, still the probability of a conclusive outcome, when using this procedure, is less than half. }%
\bibitem{endnote10}{The condition  $p_2/p_1 \geq F/(1-F)^2$ can also be written $\protect p_2/(p_1+p_2) \geq \protect F/(1-F+F^2)$ which, to a first order approximation, means that the rate of two photon pulses is larger than $F$.}%
\end{thebibliography}
\end{document}